\documentclass[12pt]{article}
\usepackage{natbib}
\usepackage{amsmath}
\usepackage{amsthm}
\usepackage{amssymb}
\usepackage{amsthm}
\usepackage{graphicx}
\usepackage{sidecap}
\usepackage{subcaption}
\usepackage{caption}
\usepackage{comment}
\usepackage{float}
\usepackage{tikz}
\usepackage{pgf, pgfplots}
\pgfplotsset{compat=1.17} 
\usepackage{xcolor}
\usetikzlibrary{arrows}
\usetikzlibrary{calc}
\usetikzlibrary{positioning}
\usetikzlibrary{automata}
\usetikzlibrary{shapes}
\usetikzlibrary{shapes.geometric}
\usetikzlibrary{decorations.markings}
\usepackage{setspace}
\usepackage{fullpage}
\usepackage{pdfpages}
\usepackage{import}
\usepackage{rotating}
\usepackage{hyperref}
\usepackage{cleveref} 
\usepackage{multirow}

\newtheorem*{thm*}{Theorem}
\newtheorem{prop}{Proposition}
\newtheorem*{prop*}{Proposition}

\newtheorem*{lem*}{Lemma}
\newtheorem{defn}{Definition}
\newtheorem{rem}{Remark} 
\newtheorem*{rem*}{Remark}

\usepackage{comment}
\definecolor{green}{HTML}{51A351}
\definecolor{blue}{HTML}{2F96b4}
\definecolor{gray}{HTML}{BCBCBC}

\DeclareMathOperator*{\argmax}{arg\max}

\title{Why do elites extend property rights: unlocking investment and the switch to public goods} 
\author{Alastair Langtry\footnote{University of Bristol.  Email: \emph{alastair.langtry@bristol.ac.uk.} I am grateful to Toke Aidt, Jean Paul Carvalho, Matthew Elliott, Lukas Freund, Gilat Levy, Marcel Schlepper, Julia Shvets, Ronny Razin, Sarah Taylor, Stephane Wolton, and Alfred Zhang for insightful comments and discussions, and to seminar participants at the University of Cambridge, the LSE Young Theorists Conference and the Silvaplana Workshop on Political Economy. This work was supported by the Economic and Social Research Council [award reference ES/P000738/1]. Any remaining errors are the sole responsibility of the author.}}

\date{\today \\ \vspace{7mm}}

\begin{document}

\maketitle
\doublespacing
\begin{abstract}
This paper presents a new rationale for a self-interested economic elite voluntarily extending property rights. When agents make endogenous investment decisions, there is a commitment problem. Ex-post, the elite face strong incentives to expropriate investments from the non-elite (who don’t have property rights), which dissuades investment. Extending property rights to new groups can resolve this problem, even for those not given property rights, by making public good provision more attractive to the elite. Unlike other models of franchise extensions, extending property rights in this paper does not involve the elite ceding control to others. Rather, it changes the incentives they face. Additionally, adding identity groups to the model shows that an elite faces weaker incentives to resolve the commitment problem when it is part of a minority identity -- identity fragmentation makes it harder for a society to extend property rights.
\end{abstract}

\onehalfspacing
\newpage
The presence of well-enforced property rights is a standard assumption in most economic models. But in practice, many people do not have secure property rights, and for much of history they have only been available to economic elites. Their spread has been a key factor in economic progress. Perhaps the leading explanation for why they have been extended is as a response to the threat of revolution \citep{acemoglu2000did, acemoglu2006economic}. There, it is a 
commitment to let the poor choose policy -- which is necessary to stave off revolution. Another important explanation focuses on how extending voting rights affects the incentives of internally warring elites \citep{lizzeri2004did, llavador2005partisan}. 

This paper presents an alternative mechanism. It shows that extending property rights can resolve a hold-up problem when a traditional commitment device is unavailable, and induces everyone else to invest costly effort into production. Critical to the mechanism is the presence of a public good. Because it is non-rival, providing a public good is relatively more attractive compared to expropriation when the elite is larger. So by increasing the size of the elite enough, the elite change their own incentives towards public good provision. The elite never cede power or tie their own hands. Rather, they change what they \emph{prefer}.

The key insight of the paper is the new mechanism: it explains why a self-interested elite might extend property rights, even absent any internal disagreements or external threats. Additionally, it shows that the initial size of the elite can itself be an important factor in whether the elite gets expanded further. And its focus on property rights, rather than voting rights, helps show how better institutional arrangements can be created without voting and without the elite ceding power. An extension that builds in identity groups (be they ethnic, religious, or something else) also shows that elites from a minority identity group are less inclined to extend property rights. Highly fragmented identity groups can hinder the move to better institutional arrangements. Taken together, these features and insights perhaps make the model a better fit for Post-War developing economies, rather than 19th Century Europe which motivates much of the other work in this area. 

In the model, agents are initially either in the elite or they are disenfranchised. Everyone is endowed with some raw materials. But only the elite have property rights -- they can expropriate the disenfranchised, but cannot themselves be expropriated. First, the (initial) elite can choose to give property rights to extra people -- expanding the elite. Second, everyone chooses how much effort to invest in producing goods. Finally, the elite choose whether to expropriate the disenfranchised and consume the resources, or to put the resources into a public good technology. Note that this commits the elite to an extension of property rights, but leaves them free to choose whether or not to expropriate.

With a simple model comes simple mechanics. The elite's choice in the third stage depends on how widespread property rights are. If the elite is large, the spoils of expropriation must be shared widely, and they are better off putting the resources into the public good technology -- as it is non-rival. Otherwise, the elite expropriates. In the second stage, the elite always invests because they have secure property rights. The disenfranchised only invest if they expect public good provision in the third stage -- there is no point in investing if the elite will expropriate them. Conditional on second-stage decisions, extending property rights in the first stage is costly to the elite -- it leaves fewer people to expropriate. The only benefit is that it can induce investment by the disenfranchised. Therefore the elite will not extend rights further than they need to.

This boils down to two institutional arrangements: one where the elite extends property rights to some new people -- but not to everyone -- and invests in public goods, and one where it does not and continues to simply expropriate. We can think of these as `inclusive' and `extractive' institutions.\footnote{This is the terminology adopted in \cite{acemoglu2012nations}. While inclusive institutions in my model involve some agents without secure property rights, they are inclusive in the sense that there is public good provision from which everyone benefits equally.} 

I provide an exact characterisation of when the elite wants to extend property rights and provide public goods, and when it does not. Notably, the elite is more inclined to extend property rights when the elite is larger to start with. This is because it is less costly to switch the elite's preferred option from expropriation to public good provision, as a smaller expansion in property rights is needed to make this switch. Unsurprisingly, the elite is more inclined to extend property rights when investment and/or the public good technology is more productive (there is more to be gained from incentivising investment and using it to provide public goods). And they are less inclined when there is a greater endowment of goods that are not produced with effort (there is more to steal). 

Finally, I consider an extension where people belong to identity groups and have some in-group altruism. That is, they care about those who share their identity. The key insight is that, within the confines of my model, societies find it easier to improve institutions and get widespread investment when a majority group is in power. This is because the elite cares more about the disenfranchised -- as a larger fraction of the ruling identity group starts out in the disenfranchised group. This can push the elite to adopt better institutions because doing so unambiguously improves outcomes for the disenfranchised.

\paragraph{A brief example.} 
To help fix ideas, consider how this mechanism may have played a role in the economic reforms in the People's Republic of China that started in the late 1970s. These reforms, known as ``reform and opening-up'' (\emph{G\v{a}ig\'{e} k\={a}if\`{a}ng}) within China, were started at the end of 1978 by Deng Xiaoping 
and are widely recognised as helping to bring about China's growth in subsequent decades \citep{ross2018china}.

Before the reforms, individuals outside the small core of communist party leadership had few, if any, property rights. Central planning did not let individuals retain the proceeds of their investments. One of the first major reforms was the decollectivisation (and in many cases de facto privatisation) of farms and the creation of non-farm `township and village enterprises' in rural areas. This extended property rights to new groups by allowing them to engage in some market transactions. It was also coupled with greater security for people owning property \citep{huang2008capitalism}. At the same time, the communist party engaged in a substantial expansion of public goods provision -- notably transport infrastructure and legal codes.
These institutional changes are widely credited with the large increase in incomes that accompanied them \citep{ray2002chinese}. 

Important features of the reforms -- an extension of property rights, public good provision and significant economic growth -- all line up with the core features of my model. But expanding rights and improving institutions are not special to my model. Two key features do differentiate my model from others. First is the kind of rights that were extended: China's reforms never involved an extension of \emph{voting} rights. This is exactly the case in my model, but something other models would struggle to accommodate. Second, is the \emph{motivation} for institutional change. My model claims the elite pursues reform in order to reap the benefits of economic growth. This appears to be the actual motivation of communist party leadership at the time \citep{ray2002chinese, aiguo2016china, gao2022china}. While anecdotal, one of Deng Xiaoping's more famous quotes is telling: \emph{``Poverty is not socialism. To be rich is glorious.''}  
Together, these features suggest my model and the economic reforms in China are a useful match.

\subsection*{Related literature}\label{sec:literature}
\paragraph{Extending Rights.} This paper principally sits within a strand of literature that asks why economic elites extend rights to new groups. To date, there are two main views. The first, following \cite{acemoglu2000did}, is that elites extend voting rights in order to preempt the threat of (violent) revolution. Extending voting rights cedes control over policy decisions to the newly enfranchised group -- it is a commitment device.\footnote{This core idea has been extended in a number of ways: \cite{conley2001endogenous} consider how the disenfranchised poor can invest in \emph{creating} a threat of revolution, and \cite{pittaluga2015democracy} consider a threat coming from an enfranchised middle class (rather than the disenfranchised poor).} 
The second, following \cite{lizzeri2004did} and \cite{llavador2005partisan}, is that it is a response to \emph{intra}-elite conflict. Some faction within the elite benefits from expanding the set of voters. 

My main contribution here is to provide a new explanation for why the elite would extend rights -- they want to increase the size of the pie by incentivising investment by the disenfranchised.  I also focus on property rights, as distinct from voting rights. In my model, the newly enfranchised need not have any say over policy. Extending rights changes the incentives, and hence choices, of the initial elite. This mechanism is closest in spirit to \cite{lizzeri2004did}, where extending rights also changes the elite's preferences. Both also emphasise the key role of public goods and their ability to facilitate better institutions.

\paragraph{Public goods, ethnic groups \& institutions.} Existing work in this strand of literature mostly focuses on how political institutions affect public good provision \citep{lake2001invisible, de2005logic, deacon2009public} -- finding that democracies provide more public goods. I examine the relationship in both directions. I recover the finding that more widely shared rights lead to more public goods. But I also find that more productive public goods can lead to more widely shared rights.

In line with a strand of literature that considers the impact of ethnic group makeup on institutions, I find that a highly fragmented ethnic group structure is bad for public good provision and can hinder institutional improvements. Existing work focuses on heterogeneous valuation of public goods by different ethnic groups \citep{alesina1999public, besley2011pillars, bandiera2011diversity} as the driving force. My mechanism is different -- it is driven by in-group altruism. Elites are more willing to provide public goods when more of their own ethnic group is outside the elite because public goods benefit everyone.

\paragraph{Bandit problems.} Finally, by bringing production to the foreground, my model also connects to work on `bandit problems'. There, ``theft by `roving bandits' destroys incentives to invest and produce'' \citep{olson1993dictatorship}. Everyone can be better off when the bandit settles down to become a `stationary bandit' and provides the public good ``peace''. However, these models typically lack a mechanism for moving between institutional regimes, and cannot explain why a bandit would extend rights to others. My model does both.\footnote{They also bear some similarity to an extreme form of a hold-up problem \citep{grossman1986costs, tirole1986procurement}. The non-elite make investments ex-ante, and then the elite (or the bandit) can expropriate the surplus ex-post. The contracting problem -- the inability to commit not to expropriate -- is key both here and in the classic hold-up literature. But my model has no notion of joint production or relationship-specific investments. Consequently, the solution to the problem is very different. Transfers of ownership (a typical recommendation from classic hold-up work) have no place here.}

\section{Model}\label{sec:model}
\paragraph{Agents and endowments.} There are $N$ agents, with typical agent $i$. Agents are either \emph{elite} or \emph{disenfranchised}. Initially, there are $E^0$ elite agents and $(N - E^0)$ disenfranchised agents. Elite agents have property rights. Disenfranchised agents do not. Having property rights means that an agent cannot have her resources taken from her, and receives an equal share of resources taken from others. For convenience, I will use subscripts $d$/$e$ to denote typical disenfranchised/elite agents. Agents are otherwise identical. 

Each agent $i$ is endowed with \emph{raw materials}, $M>0$. They can also \emph{invest} $I_i$ units of effort to produce $A f(I_i)$ units of \emph{finished goods}, with $A > 0$, $f(\cdot)$ strictly increasing and concave, and $f(0) = 0$. Also assume that $A f'(0) > 1$ and $A f'(\infty) < 1$. I will refer to raw materials and finished goods collectively as \emph{resources}, $R_i \equiv A f(I_i) + M$.

\paragraph{Timing \& Decisions.} The game has three stages. In stage one, the (initial) elite choose whether to grant property rights to some disenfranchised agents, and if so, how many. In stage two, agents make investment decisions. In stage three, the (new) elite choose whether to steal from the disenfranchised or to provide public goods.

\emph{Stage One.} One initial elite agent is chosen uniformly at random. This agent chooses whether to grant property rights to some disenfranchised agents. Formally, she chooses $E \in [E^0, N]$.\footnote{All members of the initial elite have the same preferences, so having one agent make this choice is just to avoid the equilibrium multiplicity issues common to voting games (i.e. where agents are indifferent between all choices because they are not pivotal). The same reasoning also applies to the third stage.} For convenience, I allow this extension to be to fractions of an agent. Property rights cannot be taken away, hence we must have $E \geq E^0$. If $E > E^0$, then $(E - E^0)$ initially disenfranchised agents are selected uniformly at random to join the elite.

\emph{Stage two.} All agents simultaneously choose a level of investment $I_i \geq 0$. 

\emph{Stage three.} One elite agent is chosen uniformly at random. This agent chooses whether to \emph{steal} from (i.e. expropriate) disenfranchised agents ($S$) or use the disenfranchised agents' resources to provide \emph{public} goods ($P$).\footnote{We could add an option to do nothing and leave the disenfranchised untouched. But it is clear this strategy is dominated and so will never be chosen.} 
Denote the choice $V \in \{S,P\}$. Assume the elite agent resolves indifference in favour of public good provision.

If the elite steal, disenfranchised agents' resources are split equally amongst the elite. If the elite provide public goods, disenfranchised agents' resources are used to produce a public good with a linear technology. Using $x$ units of resources creates $G \cdot x$ units of public goods, with $G \in (1/N , 1)$. 

\paragraph{Strategies.} Initial elite agents make three decisions (one in each stage). So a strategy for an elite agent is $\sigma_e = (E_e, I_e, V_e)$, where $E_e \in \mathbb{R}$, $I_e: \mathbb{R} \to \mathbb{R}$, and $V_e: \mathbb{R} \to \{S,P\}$.\footnote{Formally, the mapping for $I_e$ should be  $I_e: \mathbb{R}^{E^0} \to \mathbb{R}$. In principle, $I_e$ could be a function of both $E$ and the identity of the agent who made the choice. But only $E$ matters for payoffs, and so is a sufficient statistic for the outcome of stage one. Similarly, the mapping for $V_e$ should be $V_e: \mathbb{R}^{E^0} \times \mathbb{R}^N \to \{S,P\}$. But total resources held by the non-elite is a sufficient statistic for the outcome of stage two.}

Initially disenfranchised agents face two possible options. Either they remain disenfranchised -- in which case they only choose investment in the second stage. Or they join the elite, in which case they choose investment in the second stage and (might) choose the institutional arrangement in the third. So a strategy for an initially disenfranchised agent is $\sigma_d = (I_d, I_e,V_e)$, where $I_d: \mathbb{R} \to \mathbb{R}$, and $I_e,V_e$ are as before (this is because, once they have joined the elite, they become identical to the initial elite agents).

\paragraph{Payoffs.} Agents care only about their own material payoffs. All agents of a given type are identical, so we index payoffs by type. The payoffs are: 
\begin{align}\label{eq:payoffs}
\Pi_d = 
\begin{cases}
(N - E) R_d G - I_d \\ 
- I_d               
\end{cases}
    \ \text{ and } \quad 
\Pi_e = 
\begin{cases}
(N -prefersd G + R_e - I_e \quad &\text{ if } V = P \\
(N - E) R_d \frac{1}{E} + R_e - I_e \quad &\text{ if } V = S
\end{cases}
\end{align}       

\paragraph{Solving the Game.} The game is one of complete information. So we will look for Pure Strategy Subgame Perfect Nash Equilibria (\emph{equilibria}, for convenience) of this game, and solve by Backwards Induction.

\paragraph{An assumption.} Finally, assume the initial elite is not too large. Specifically: $E^0 < 1/G$. This is only for clarity of exposition. When the assumption does not hold, then the elite prefer to provide public goods even without extending the franchise: there is no commitment problem. It is an uninteresting case. For completeness, I cover it in the online appendix.

\paragraph{Discussion.}
The two key assumptions in this model are: (1) that the elite lack a traditional commitment device, and (2) that all members of the elite get a share of expropriated resources. The first creates the potential for a hold-up problem; the disenfranchised might be dissuaded from investment by fear of expropriation. The second ensures that expropriation is relatively less attractive when the elite is larger (compared to public good provision).

\begin{enumerate}
    \item \textbf{No commitment device:} this requires a world where the elite cannot create rules that constrain their future actions -- they always have the power to change rules when they want to. This is a common assumption in games with a hold-up style problem (for example, \cite{acemoglu2000did}). Mechanically, it is reflected in the timing of the game (but note that the insights would not change if investment and expropriation decisions happened simultaneously).
    
    \item \textbf{Sharing in expropriated resources:} every member of the elite must have enough power to get \emph{some} constant share of expropriated resources, including those who were recently granted property rights. One way to view this, natural in some contexts, is that a new member of the elite is granted some land that she will control, and from which she can extract rents. Alternatively, we can view this as the outcome of some unmodelled bargaining or contest process through which the elite divide up the expropriated resources. The stronger assumption that every member gets an \emph{equal} share is helpful for tractability and for presenting clean results (in particular, it makes members of the elite homogeneous), but the qualitative insight would remain if it were relaxed.
\end{enumerate}

Additionally, I have assumed that everyone benefits equally from the public good. This is a natural benchmark, but it is not important. As with the elites sharing in the expropriated resources, what matters is that everyone receives \emph{some} benefit.

\section{Results}\label{sec:results}
At the heart of the model is a commitment problem.\footnote{While a commitment problem is also at the heart of the \cite{acemoglu2000did} framework, the way it is solved there is different. There, extending the franchise involves the rich (elite) ceding decision-making power to the poor (disenfranchised). This allows the rich to commit ex-ante to a policy they would not choose ex-post.} 
Conditional on the size of the pie -- which depends critically on investment decisions -- the elite wants to take as much of it as possible. This means they will expropriate. But the disenfranchised know this, and so do not invest. Investment in the economy is low and the pie is small. A promise to provide public goods would not be credible -- in my model, the elite holds all of the power and can ultimately do whatever it wants. So to unlock investment by the disenfranchised, the elite must change what it \emph{wants} to do. This is not a standard commitment device: the elite can never tie their hands in this model. Rather, they can change the incentives they face, and so change their preferred action.

The elite can do this by extending property rights to new people (i.e. adding them to the elite) -- the one thing they can commit to in my model -- which makes expropriation less attractive relative to providing public goods.\footnote{This commitment power is embedded in the inability to reverse an extension of property rights.} 
With a larger elite, there are more people who take a share of any expropriated resources. In contrast, the public good is non-rival, so this cost of a larger elite does not apply when the elite provides public goods. Therefore, if enough extra people are given property rights, expropriation becomes less attractive in \emph{absolute} terms than providing public goods. This resolves the commitment problem. We can see this most clearly in the elite's payoffs in \Cref{eq:payoffs}, which I restate here for convenience:
\begin{align*}
    \Pi_e(V=P) &= (N - E) R_d G + R_e - I_e, \\
    \Pi_e(V=S) &= (N - E) R_d \frac{1}{E} + R_e - I_e.
\end{align*}

It is clear that $\Pi_e(V=P) \geq \Pi_e(V=S)$ if and only if $E \geq 1/G$.
However, conditional on $R_d$, both $\Pi_e(V=S)$ and $\Pi_e(V=P)$ are decreasing in $E$. So the elite never want to extend property rights by more than is necessary to resolve the commitment problem and induce investment by the disenfranchised. This is because extending property rights leaves fewer disenfranchised agents whose resources can be used to provide public goods, and means they must share more widely. All else equal, the elite benefit from having a large disenfranchised group because it leaves them with more people to exploit. So if they occur, extensions of property rights will only be partial.

\begin{rem}
    Property rights will never be extended to everyone.
\end{rem}

Another implication of this logic (and the highly stylised setup of my model) is that there is no benefit in extending property rights to some people if the extension is not sufficient to resolve the commitment problem.

This leaves two possible outcomes. In one, the elite extends property rights just enough to change their own preferences over the institutional arrangement -- which is to exactly $E = 1/G$. They then choose inclusive institutions (i.e. to provide public goods) and the disenfranchised invest. In the other, the elite do not extend the franchise at all and choose extractive institutions. This comes at the cost of lower investment and hence a smaller pie -- but the elite get a bigger slice of it. 

Which institutional arrangement the elite chooses comes down to a simple comparison of their own material payoffs under each scenario. This creates a sharp threshold. On one side, the elite prefers to extend property rights and adopt inclusive institutions. On the other, they choose extractive institutions and forego investment by the disenfranchised. The following result formalises this discussion.

\begin{prop}\label{prop:eq}
There is a unique equilibrium. The elite extends property rights to $E = 1/G$, provide public goods, and all agents invest if and only if
\begin{align}\label{eq:threshold}
    E^0 \cdot \left(G + (G - \frac{1}{N}) \cdot \frac{A f(I_{d}^*)}{M} \right) \geq 1,
\end{align}
where $I_d^* = f^{\prime -1}(1/GA)$.
Otherwise, the elite does not extend property rights, expropriates, and only elite agents invest.
\end{prop}
 
Comparative statics follow immediately from this characterisation of equilibrium behaviour. An increase in the initial size of the elite, $E^0$, means that the elite have to share expropriated resources more widely, which reduces the benefits of expropriation. Higher productivity of individual effort, $A$, raises the resources available if the disenfranchised can be induced to invest, increasing the benefits of resolving the commitment issue. Both make inclusive institutions relatively more attractive, and so make the elite more inclined to extend property rights.
In contrast, a larger endowment of raw materials, $M$, makes adopting inclusive institutions less attractive. This is because there is more for the elite to steal absent investment by the disenfranchised, so each elite agent must give up more when a new agent is granted property rights. 

The role of the productivity of the public goods technology is more nuanced. By increasing the value of spending on public goods, and hence the payoff to resolving the commitment problem, it makes adopting inclusive institutions more attractive. But in doing so it makes provision of public goods the preferred option, and hence credible, for a smaller size of the elite. And as discussed, the elite never extends property rights further than needed to resolve the commitment issue. So a more productive public goods technology makes an extension of property rights more likely, but makes it smaller when it does happen.

\begin{prop}\label{prop:comp stat G}
The size of the elite is decreasing in $G$ if \Cref{eq:threshold} is satisfied, and is increasing in $G$ otherwise.
\end{prop}

I now turn to comparative statics for total output -- essentially the level of economic activity in the society. Total output (which I will denote $Y$) is the sum of all agents' payoffs: $Y = E \Pi_e + (N - E) \Pi_d$. In this model, total output coincides with utilitarian welfare, and so we could view results about output as reading across to welfare.\footnote{But this coincidence is an artefact of linear preferences, so we should be careful not to put too much weight on a welfare interpretation.}

In line with basic intuition, total output is increasing in both the initial size of the elite, $E^0$, and in the productivity of individual effort, $A$. While mostly straightforward, there are two forces at work: (1) the direct impact on output for a given institutional arrangement, and (2) the impact on which institutional arrangement is chosen. For these parameters, both effects push in the same direction.

For the endowment of raw materials, $M$, the two effects work against each other. Conditional on the institutional arrangement, more raw materials are clearly good -- they add directly to welfare. But more raw materials can induce the elite to steal, where they might otherwise have extended property rights and provided public goods. So a small increase in $M$ increases welfare everywhere \emph{except} at the point where it induces the elite to switch from providing public goods to stealing. 

The impact of the public good productivity, however, is ambiguous. It has a direct impact on output -- higher $G$ increases output for any given level of inputs. And it also makes the initial elite more inclined to adopt inclusive institutions. But it also makes it easier to resolve the commitment problem -- so when the initial elite do extend property rights, they do so by less. This third effect may dominate. For some parametrisations, having fewer agents in the elite (all else equal) reduces total output. 

\paragraph{Adding growth.} The game here is one-shot, and so cannot have a meaningful notion of growth. But it is easy to play this game infinitely many times and to enrich it with (endogenous) learning-by-doing technology growth a la \cite{romer1986increasing} -- where the next period's TFP depends on investment today.\footnote{For simplicity, I assume that agents do not account for the impact their decisions (either over $I$ or over $E$) have on future technology -- so agents treat future technology as exogenous. This shuts down dynamic considerations regarding investment and regarding the institutional arrangement.} 
In this world, technology initially grows slowly due to the investment by the elite and is on a path to converge to some low steady state. But if that technology growth pushes society over the threshold set out in \Cref{prop:eq} then there is institutional change -- which unlocks investment by the disenfranchised. This higher level of investment in turn spurs more technology growth, and leads society to converge to a high(er) steady state.

Whether a society ever gets over the threshold from \Cref{prop:eq} depends critically on the other parameters. More valuable raw materials and a smaller initial elite make the threshold higher, and so make the society more likely to get stuck in the low steady state. This yields political economy versions of the standard poverty trap and resource curse results, now driven by my mechanism. It also creates a symbiotic relationship between technology and institutions -- a society can have bad technology because it has bad institutions, and also have bad institutions because it has bad technology. I set out this extended model and the associated results formally in the Online Appendix.

\section{Identity Groups}\label{sec:extension}
I now extend the model by endowing agents with a \emph{group} -- this could be an ethnic, religious or identity group -- and assuming some in-group altruism. Such a preference for the in-group finds substantial support in the experimental literature (e.g. \cite{bernhard2006group, yamagishi2008does, abbink2012parochial} and references within). This allows us to explore how the social makeup of a society affects its institutional arrangements. The headline result is that elites from a majority group are more likely to extend property rights and choose inclusive institutions. However, when they do so, they will add fewer others into the elite. This is because their altruism makes public good provision more attractive, and consequently easier to resolve the commitment problem.

\paragraph{Extending the model.} The model is as described in \Cref{sec:model}, with the following change: each agent has an exogenous group $j \in J$, and they put some weight on the average material payoff of the other agents in group $j$.\footnote{While there are many ways that group membership could affect preferences, I focus on in-group altruism. Additionally, I assume that people identify only with one exogenous group. This abstracts away from the fact that people may identify with multiple groups, and can \emph{choose} which to make salient. A growing literature, following \cite{akerlof2000economics} and \cite{shayo2009model}, explores these identity choices. More recent work includes \cite{atkin2021we, grossman2021identity} and \cite{ghiglino2024endogenous}.}
Formally, preferences are: 
\begin{align}\label{eq:payoffs tribe}
    U^j_i = \Pi_i + \alpha (q^j \Pi_e + (1 - q^j) \Pi_d) \quad \text{for} \ i \in \{d,e\},
\end{align}
where $q_j$ is the fraction of group-$j$ agents who are in the elite and $\alpha \in (0,1)$ captures altruism towards the in-group. As before, $\Pi_e$ and $\Pi_d$ are the material payoffs to an agent elite/disenfranchised respectively. Without loss of generality, assume that agents in group $j$ make up a fraction $p_j^{tot}$ of the total, and a fraction $p_j^{elite}$ of the initial elite. It then follows that $q_j \equiv \frac{p_j}{p_j^{tot}} \times \frac{E}{N}$.
To keep things simple, I focus on the case where all agents in the initial elite belong to group $j$. Formally, I assume that $p_j^{elite} = 1$ for some $j \in J$ and $p_{j'}^{elite} = 0$ for all $j' \neq j$. This also maintains the assumption from earlier that all agents in the elite are homogeneous.

\paragraph{Results} When group $j$ is larger (relative to the total population), a greater fraction of group-$j$ agents are disenfranchised. This means the initial elite place more weight on the material payoffs of disenfranchised agents. In turn, this makes public good provision more attractive relative to expropriation. 

\begin{prop}\label{prop:bigger tribe}
Suppose $p_j^{elite} = 1$. An increase in the size of group $j$ can cause the elite to switch to providing public goods.
\end{prop}

In terms of the institutions chosen, a larger group being in power can only help. This is because more of their group are disenfranchised, so they care more about the disenfranchised. The implication of this result is stark. There will be some cases where a country has a ruling elite from a minority group and maintains extractive institutions, but where a change of the ruling elite to a majority group would bring about a change to inclusive institutions.

However, this additional altruism benefit also has a perverse effect. Because the elite care more about the disenfranchised agents when their group is larger, the elite does not need to extend property rights by as much in order to prefer public goods provision over expropriation. But this reduces the amount by which property rights must be extended to resolve the commitment problem and unlock investment by the disenfranchised. And so the elite will extend property rights by less -- while somewhat altruistic, the elite still prefers to only extend property rights just far enough to unlock investment.

\begin{prop}\label{prop:smaller tribe}
   Suppose $p_j^{elite} = 1$. And suppose the parameters are such that the elite chooses to extend property rights and provide public goods. Then the number of agents who receive property rights in this extension is strictly decreasing in the size of group $j$.
\end{prop}

Perhaps counter-intuitively, greater altruism would also have the same effect. The elite care more about the disenfranchised, making them more inclined to provide public goods. But with $\alpha < 1$, they are not \emph{sufficiently} altruistic to want to bring others into the elite just for the benefits it brings others. They need some personal benefit too -- namely unlocking investment by disenfranchised agents. So, again, the elite only adds the minimum number of people necessary to credibly commit to providing public goods. Like group size, stronger in-group altruism is not an unalloyed good.

\paragraph{Discussion.} The idea that ethnic fragmentation can be associated with worse economic outcomes is far from new, and finds support in cross-country data \citep{easterly1997africa, alesina2005ethnic}. This has been both considered and tested especially in the context of post-colonial Africa. That a causal link might be working through the impact of ethnic fragmentation on the prevailing institutions has been suggested in \citet[esp. S7]{acemoglu2010africa} -- although they do not present a formal model. 

The two countries they discuss in detail -- Sierra Leone and Botswana -- are instructive here. Sierra Leone had poor economic outcomes in the forty years following its independence from the United Kingdom in 1961, with particularly poor provision of public goods \citep[p.42]{acemoglu2010africa} -- and is an ethnically fragmented country, with no ethnic majority. In contrast, Botswana is one of Africa's post-colonial success stories -- and also one of its most ethnically homogeneous \citep{CIA_botswana, CIA_sierra}. My model provides a simple mechanism for how this might work as a causal relationship.

\section{Conclusion}\label{sec:conclusion}
Existing work provides two main reasons for elites to extend the franchise: to avoid revolution, or to resolve internal disagreements. This paper suggests a more simplistic one: doing so might be profitable for the elite. In the absence of a commitment device, extending rights can allow the elite to unlock investment by the disenfranchised. This works by changing the preferences of the elite, pushing them towards more inclusive institutional arrangements, which in turn changes the investment decisions of the disenfranchised. 

It differs from existing work on franchise extensions in two further ways. First, the elite resolve a commitment problem by changing what they want, rather than by tying their hands and letting others choose. In my model, the elite always retain power. Second, it focuses on the extension of \emph{property} rights, rather than of \emph{voting} rights. An analysis of property rights, as distinct from voting rights, may be important when thinking about institutional change in countries that either do not hold elections or where elections are tokenistic. 
Additionally, it examines the role that identity groupings (whether along ethnic, religious, or some other lines) can play in determining institutional arrangements. Within the model, having highly fragmented identity groups hinders the extension of property rights. This paper mainly aims fit to Post-War and contemporary developing economies, rather than the waves of institutional change in 19th-century Europe that motivated existing work on franchise extensions.






\singlespacing
\bibliographystyle{abbrvnat}
\addcontentsline{toc}{section}{References}
\bibliography{bib}

\begin{thebibliography}{37}
\providecommand{\natexlab}[1]{#1}
\providecommand{\url}[1]{\texttt{#1}}
\expandafter\ifx\csname urlstyle\endcsname\relax
  \providecommand{\doi}[1]{doi: #1}\else
  \providecommand{\doi}{doi: \begingroup \urlstyle{rm}\Url}\fi

\bibitem[Abbink et~al.(2012)Abbink, Brandts, Herrmann, and
  Orzen]{abbink2012parochial}
K.~Abbink, J.~Brandts, B.~Herrmann, and H.~Orzen.
\newblock Parochial altruism in inter-group conflicts.
\newblock \emph{Economics Letters}, 117\penalty0 (1):\penalty0 45--48, 2012.

\bibitem[Acemoglu and Robinson(2000)]{acemoglu2000did}
D.~Acemoglu and J.~A. Robinson.
\newblock Why did the west extend the franchise? democracy, inequality, and
  growth in historical perspective.
\newblock \emph{The Quarterly Journal of Economics}, 115\penalty0 (4):\penalty0
  1167--1199, 2000.

\bibitem[Acemoglu and Robinson(2006)]{acemoglu2006economic}
D.~Acemoglu and J.~A. Robinson.
\newblock \emph{Economic origins of dictatorship and democracy}.
\newblock Cambridge University Press, 2006.

\bibitem[Acemoglu and Robinson(2010)]{acemoglu2010africa}
D.~Acemoglu and J.~A. Robinson.
\newblock Why is africa poor?
\newblock \emph{Economic history of developing regions}, 25\penalty0
  (1):\penalty0 21--50, 2010.

\bibitem[Acemoglu and Robinson(2012)]{acemoglu2012nations}
D.~Acemoglu and J.~A. Robinson.
\newblock \emph{Why nations fail: The origins of power, prosperity, and
  poverty}.
\newblock Crown Business, 2012.

\bibitem[Aiguo(2016)]{aiguo2016china}
L.~Aiguo.
\newblock \emph{China and the Global Economy since 1840}.
\newblock Springer, 2016.

\bibitem[Akerlof and Kranton(2000)]{akerlof2000economics}
G.~A. Akerlof and R.~E. Kranton.
\newblock Economics and identity.
\newblock \emph{The Quarterly Journal of Economics}, 115\penalty0 (3):\penalty0
  715--753, 2000.

\bibitem[Alesina and La~Ferrara(2005)]{alesina2005ethnic}
A.~Alesina and E.~La~Ferrara.
\newblock Ethnic diversity and economic performance.
\newblock \emph{Journal of economic literature}, 43\penalty0 (3):\penalty0
  762--800, 2005.

\bibitem[Alesina et~al.(1999)Alesina, Baqir, and Easterly]{alesina1999public}
A.~Alesina, R.~Baqir, and W.~Easterly.
\newblock Public goods and ethnic divisions.
\newblock \emph{The Quarterly journal of economics}, 114\penalty0 (4):\penalty0
  1243--1284, 1999.

\bibitem[Arrow(1962)]{arrow1962economic}
K.~J. Arrow.
\newblock The economic implications of learning by doing.
\newblock \emph{The Review of Economic Studies}, 29\penalty0 (3):\penalty0
  155--173, 1962.

\bibitem[Atkin et~al.(2021)Atkin, Colson-Sihra, and Shayo]{atkin2021we}
D.~Atkin, E.~Colson-Sihra, and M.~Shayo.
\newblock How do we choose our identity? a revealed preference approach using
  food consumption.
\newblock \emph{Journal of Political Economy}, 129\penalty0 (4):\penalty0
  1193--1251, 2021.

\bibitem[Bandiera and Levy(2011)]{bandiera2011diversity}
O.~Bandiera and G.~Levy.
\newblock Diversity and the power of the elites in democratic societies:
  Evidence from indonesia.
\newblock \emph{Journal of Public Economics}, 95\penalty0 (11-12):\penalty0
  1322--1330, 2011.

\bibitem[Bernhard et~al.(2006)Bernhard, Fehr, and
  Fischbacher]{bernhard2006group}
H.~Bernhard, E.~Fehr, and U.~Fischbacher.
\newblock Group affiliation and altruistic norm enforcement.
\newblock \emph{American Economic Review}, 96\penalty0 (2):\penalty0 217--221,
  2006.

\bibitem[Besley and Persson(2011)]{besley2011pillars}
T.~Besley and T.~Persson.
\newblock Pillars of prosperity.
\newblock In \emph{Pillars of Prosperity}. Princeton University Press, 2011.

\bibitem[{CIA World Factbook}(2024{\natexlab{a}})]{CIA_botswana}
{CIA World Factbook}.
\newblock Country summary: Botswana, 2024{\natexlab{a}}.
\newblock URL
  \url{https://www.cia.gov/the-world-factbook/countries/botswana/summaries}.

\bibitem[{CIA World Factbook}(2024{\natexlab{b}})]{CIA_sierra}
{CIA World Factbook}.
\newblock Country summary: Sierra leone, 2024{\natexlab{b}}.
\newblock URL
  \url{https://www.cia.gov/the-world-factbook/countries/sierra-leone/summaries/}.

\bibitem[Conley and Temimi(2001)]{conley2001endogenous}
J.~P. Conley and A.~Temimi.
\newblock Endogenous enfranchisement when groups’ preferences conflict.
\newblock \emph{Journal of Political Economy}, 109\penalty0 (1):\penalty0
  79--102, 2001.

\bibitem[De~Mesquita et~al.(2005)De~Mesquita, Smith, Siverson, and
  Morrow]{de2005logic}
B.~B. De~Mesquita, A.~Smith, R.~M. Siverson, and J.~D. Morrow.
\newblock \emph{The logic of political survival}.
\newblock MIT press, 2005.

\bibitem[Deacon(2009)]{deacon2009public}
R.~T. Deacon.
\newblock Public good provision under dictatorship and democracy.
\newblock \emph{Public choice}, 139\penalty0 (1):\penalty0 241--262, 2009.

\bibitem[Easterly and Levine(1997)]{easterly1997africa}
W.~Easterly and R.~Levine.
\newblock Africa's growth tragedy: policies and ethnic divisions.
\newblock \emph{The Quarterly Journal of Economics}, pages 1203--1250, 1997.

\bibitem[Gao(2022)]{gao2022china}
X.~Gao.
\newblock Reform and opening-up to economic miracle, 2022.
\newblock URL
  \url{https://www.chinadaily.com.cn/a/202210/19/WS634f37e3a310fd2b29e7d394.html}.

\bibitem[Garnaut et~al.(2018)Garnaut, Song, and Fang]{ross2018china}
R.~Garnaut, L.~Song, and C.~Fang, editors.
\newblock \emph{China’s 40 Years of Reform and Development: 1978–2018}.
\newblock ANU Press, 2018.
\newblock ISBN 9781760462246.
\newblock URL \url{http://www.jstor.org/stable/j.ctv5cgbnk}.

\bibitem[Ghiglino and Tabasso(2024)]{ghiglino2024endogenous}
C.~Ghiglino and N.~Tabasso.
\newblock Endogenous identity in a social network.
\newblock \emph{arXiv preprint arXiv:2406.10972}, 2024.

\bibitem[Grossman and Helpman(2021)]{grossman2021identity}
G.~M. Grossman and E.~Helpman.
\newblock Identity politics and trade policy.
\newblock \emph{The Review of Economic Studies}, 88\penalty0 (3):\penalty0
  1101--1126, 2021.

\bibitem[Grossman and Hart(1986)]{grossman1986costs}
S.~J. Grossman and O.~D. Hart.
\newblock The costs and benefits of ownership: A theory of vertical and lateral
  integration.
\newblock \emph{Journal of political economy}, 94\penalty0 (4):\penalty0
  691--719, 1986.

\bibitem[Huang(2008)]{huang2008capitalism}
Y.~Huang.
\newblock \emph{Capitalism with Chinese characteristics: Entrepreneurship and
  the state}.
\newblock Cambridge University Press, 2008.

\bibitem[Lake and Baum(2001)]{lake2001invisible}
D.~A. Lake and M.~A. Baum.
\newblock The invisible hand of democracy: political control and the provision
  of public services.
\newblock \emph{Comparative political studies}, 34\penalty0 (6):\penalty0
  587--621, 2001.

\bibitem[Lizzeri and Persico(2004)]{lizzeri2004did}
A.~Lizzeri and N.~Persico.
\newblock Why did the elites extend the suffrage? democracy and the scope of
  government, with an application to britain's “age of reform”.
\newblock \emph{The Quarterly Journal of Economics}, 119\penalty0 (2):\penalty0
  707--765, 2004.

\bibitem[Ljungqvist and Sargent(2004)]{ljungqvist2004recursive}
L.~Ljungqvist and T.~J. Sargent.
\newblock Recursive macroeconomic theory, 2004.

\bibitem[Llavador and Oxoby(2005)]{llavador2005partisan}
H.~Llavador and R.~J. Oxoby.
\newblock Partisan competition, growth, and the franchise.
\newblock \emph{The Quarterly Journal of Economics}, 120\penalty0 (3):\penalty0
  1155--1189, 2005.

\bibitem[Olson(1993)]{olson1993dictatorship}
M.~Olson.
\newblock Dictatorship, democracy, and development.
\newblock \emph{American Political Science Review}, 87\penalty0 (3):\penalty0
  567--576, 1993.

\bibitem[Pittaluga et~al.(2015)Pittaluga, Cama, and
  Seghezza]{pittaluga2015democracy}
G.~B. Pittaluga, G.~Cama, and E.~Seghezza.
\newblock Democracy, extension of suffrage, and redistribution in
  nineteenth-century europe.
\newblock \emph{European Review of Economic History}, 19\penalty0 (4):\penalty0
  317--334, 2015.

\bibitem[Ray(2002)]{ray2002chinese}
A.~Ray.
\newblock The chinese economic miracle: Lessons to be learnt.
\newblock \emph{Economic and Political Weekly}, pages 3835--3848, 2002.

\bibitem[Romer(1986)]{romer1986increasing}
P.~M. Romer.
\newblock Increasing returns and long-run growth.
\newblock \emph{Journal of Political Economy}, 94\penalty0 (5):\penalty0
  1002--1037, 1986.

\bibitem[Shayo(2009)]{shayo2009model}
M.~Shayo.
\newblock A model of social identity with an application to political economy:
  Nation, class, and redistribution.
\newblock \emph{American Political Science Review}, 103\penalty0 (2):\penalty0
  147--174, 2009.

\bibitem[Tirole(1986)]{tirole1986procurement}
J.~Tirole.
\newblock Procurement and renegotiation.
\newblock \emph{Journal of Political Economy}, 94\penalty0 (2):\penalty0
  235--259, 1986.

\bibitem[Yamagishi and Mifune(2008)]{yamagishi2008does}
T.~Yamagishi and N.~Mifune.
\newblock Does shared group membership promote altruism? fear, greed, and
  reputation.
\newblock \emph{Rationality and Society}, 20\penalty0 (1):\penalty0 5--30,
  2008.

\end{thebibliography}
\appendix
\onehalfspacing

\section{Proofs}
\begin{proof}[\emph{\textbf{Proof of Proposition \ref{prop:eq}}}]

%
\textbf{Third Stage.} At this stage $E$, $I_d$ and $I_e$ are fixed. So the elite's payoffs, $\Pi_e$, from each action are:
\begin{align}
    \Pi_e(\text{steal}) &= \underbrace{(A f(I_e^*) - I_e^* + M) \vphantom{\frac{X}{X}} }_{\text{private consumption}} 
    + \underbrace{(N - E) \cdot \frac{1}{E} \cdot (A f(I_d^*) + M)}_{\text{expropriated resources}} \\
    \Pi_e(\text{public}) &= \overbrace{(A f(I_e^*) - I_e^* + M)} 
    + \underbrace{(N - E) \cdot G \cdot (A f(I_d) + M)}_{\text{public goods consumption}}
\end{align}
Almost everything cancels, and $\Pi_e(\text{public}) \geq \Pi_e(\text{steal})$ if and only if $E \geq 1/G$. Therefore $V_e^* = P$ if $E \geq 1/G$ and $V_e^* = S$ otherwise.\footnote{Were we to add an option to do nothing in the third stage, the payout would be $\Pi_e(\text{nothing}) = (A f(I_e^*) - I_e^* + M)$, which is clearly strictly dominated.}

\textbf{Second stage.} By Backwards Induction, the disenfranchised know that the elite will provide the public good if and only if $E \geq 1/G$. So if, $E \geq 1/G$ then $I^*_d = f^{\prime -1}(1/GA)$ (this is $I_d^* = \argmax [G A f(I_d) - I_d]$), because they know that all finished goods produced will be used to provide the public good. Otherwise, $I_d^* = 0$. The elite know they keep everything they produce, regardless of the decision made in the third stage. So $I_e^* = f^{\prime -1}(1/A)$ (this is $I_e^* = \argmax [A f(I_d) - I_e]$).

\textbf{First Stage.} Conditional on the disenfranchised agents decision in the second stage, the (initial) elite always prefers a smaller elite. This is clear from the elite's payoff:
\begin{align}
    \Pi_e(E| E < 1/G) &= \underbrace{(A f(I_e^*) - I_e^* + M) \vphantom{\frac{X}{X}} }_{\text{private consumption}} 
    + \underbrace{(N - E) \cdot \frac{1}{E} \cdot M}_{\text{exp. raw materials}} \tag{steal}\label{eq:elite payoff steal} \\
    \Pi_E(E| E \geq 1/G) &= \overbrace{(A f(I_e^*) - I_e^* + M)}
    + \underbrace{(N-E) \cdot G \cdot (A f(I_d^*) + M)}_{\text{public goods consumption}}  \tag{public}\label{eq:elite payoff public}
\end{align}
Therefore, only $E=E^0$ and $E= 1/G$ could possibly be optimal (as each equation is clearly decreasing in $E$, and $E= 1/G$ is the smallest \emph{feasible} value of $E$ conditional on $E \geq 1/G$). 
All that remains is to compare the elite's payoff in each of these cases. Substitute $E=E^0$ and $E= 1/G$ into the first and second equations, respectively and then compare them. Some simple rearranging shows that $\Pi_e(\text{public} | E = 1/G) \geq \Pi_e(\text{steal} | E = E^0)$ if and only if:
\begin{align}
E^0 \left( G + (G - \frac{1}{N}) \frac{A f(I_d^*)}{M} \right) \geq 1
\end{align}
\textbf{Existence and uniqueness} follow immediately from the argument above.\footnote{It follows immediately from this proof that when $E^0 \geq 1/G$, we have $\Pi_e(\text{public}) \geq \Pi_e(\text{steal})$. So the elite always provide public goods, and hence never have an incentive to extend the franchise any further.}
\end{proof}

\begin{proof}[\emph{\textbf{Proof of Proposition \ref{prop:comp stat G}}}]
When $E^0 \left( G + (G - \frac{1}{N}) \frac{A f(I_{d}^*)}{M} \right) \geq 1$, the elite extends the property rights to exactly $E = 1/G$ -- which is clearly decreasing in $G$. The elite makes no change to property rights whenever $E^0 \left( G + (G - \frac{1}{N}) \frac{A f(I_{d}^*)}{M} \right) < 1$. Finally, the marginal increase in $G$ that leads the equation to hold with equality induces an extension of property rights. This follows from \Cref{prop:eq}. 
\end{proof}

\begin{proof}[\emph{\textbf{Proof of Proposition \ref{prop:bigger tribe}}}]
By assumption, group $j$ constitutes the all of the elite (i.e. $p_j^{elite} = 1$). So we only need to consider the preferences of the group $j$ elite.
The utility (for each option; steal or provide public goods) is:
\begin{align*}
    U_e^j(S) &= \Pi_e(S) + \alpha ( q^j \Pi_e(S) + (1 - q^j) \Pi_d(S) ) \\
    U_e^j(P) &=  \Pi_e(P) + \alpha ( q^j \Pi_e(P) + (1 - q^j) \Pi_d(P) ) \\
\implies U_e^j(P) - U_e^j(S) &= \alpha (1 - q^j) [\Pi_d(P) - \Pi_d(S)] + (1 + \alpha q^j) [\Pi_e(P) - \Pi_e(S)]
\end{align*}
where $q^j \equiv (p_j^{elite} \cdot E)/(p_j^{tot} N)$, and $\Pi_i(V)$ is the payoff to agent $i$ when the elite choose option $V$. Notice that changing the size of group $j$ only affects $q^j$. Further, it is clear that $[\Pi_d(P) - \Pi_d(S)] > 0$ (since the disenfranchised get nothing if the elite steal, and something positive if the elite provide public goods). 

So if $[\Pi_e(P) - \Pi_e(S)] \geq 0$, then $[U_e^j(P) - U_e^j(S)] > 0$, so the elite must want to provide public goods.\footnote{In fact, in this case, they would want to provide public goods, even absent the altruism.} 
This is for any value of $q^j$. So changing the size of group $j$ has no impact on the institutional arrangements chosen. 
Now consider the case where $[\Pi_e(P) - \Pi_e(S)] < 0$. We have 
\begin{align*}
    \frac{d [U_e^j(P) - U_e^j(S)]}{d p_j^{tot}} = - \alpha \frac{d q^j}{d p_j^{tot}} [\Pi_d(P) - \Pi_d(S)] + \alpha \frac{d q^j}{d p_j^{tot}} [\Pi_e(P) - \Pi_e(S)],
\end{align*}
which must be strictly positive (since $d q^j / d p_j^{tot} < 0$). So increasing the size of tribe $j$, $p_j^{tot}$, can only induce a change from stealing (exclusive institutions) to providing public goods (inclusive institutions).


\end{proof}

\begin{proof}[\emph{\textbf{Proof of Proposition \ref{prop:smaller tribe}}}]
The logic here is very similar to \Cref{prop:eq}. \textbf{Third stage.} At this stage, $I_d$ and $E$ are fixed. Taking the utility function in \Cref{eq:payoffs tribe} and substituting in the payoffs from \Cref{eq:payoffs}:
\begin{align}
    U_e^j(S) &= (R_e - I_e + (N - E) R_d \frac{1}{E})(1 + \alpha q^j) + \alpha (1 - q^j) (- I_d) \\
    U_e^j(P) &=  (R_e - I_e + (N - E) R_d G)(1 + \alpha q^j) + \alpha (1 - q^j) ((N-E) R_d G - I_d) 
\end{align}
Now, we find the conditions under which the elite prefer to provide public goods in the third stage, i.e $U_e^j(P) \geq U_e^j(S)$. Some terms cancel immediately, yielding
\begin{align*}
    ((N - E) R_d G)(1 + \alpha q^j) + \alpha (1 - q^j) ((N-E) R_d G) &\geq 
    ((N - E) R_d \frac{1}{E})(1 + \alpha q^j) 
\end{align*}
Further rearranging then yields $U_e^j(P) \geq U_e^j(S) \iff E \geq \frac{1 + \alpha q^j}{G (1 + \alpha)}$
\textbf{Second stage.} Same as \Cref{prop:eq} -- disenfranchised will invest if and only if the elite provide public goods in the third stage, which in turn happens if and only if $E \geq \frac{1 + \alpha q^j}{G (1 + \alpha)}$. \textbf{First Stage.} Conditional on the disenfranchised agents' decisions in the second stage, the (initial) elite always prefer a smaller elite. This follows from the fact that both $U_e^j(P)$ and $U_e^j(S)$ are decreasing in $E$.\footnote{For $U_e^j(S)$, this is clear from the equation above. For $U_e^j(P)$, it is easier to see by rearranging the elite's utility to: $U_e^j(P) =  (N - E) R_d G)(1 + \alpha) + (1 + \alpha q^j)(R_e - I_e) - \alpha (1 - q^j) I_d$.} 
So if the elite does extend the first stage, they only extend to exactly $E = \frac{1 + \alpha q^j}{G (1 + \alpha)}$.
\end{proof}

\newpage
\begin{center}
\section*{Online Appendix}
\end{center}
\section{A Repeated Game Extension}
\subsection{Model}
The game set out in \Cref{sec:model} is now a stage game. The stage game is repeated infinitely times. Time is indexed $t=1,2...$. Investment $I_{i,t} \geq 0$ is chosen anew each period. Agents who have property rights at the end of period $t$ are the elite who have them at the start of $t+1$. So $E^0_{t+1} = E_t$, with $E^0_1$ given exogenously. We assume that technology, $A_t$, grows according to a simple learning-by-doing process: $A_{t+1} = (1 - \delta) A_t + a \bar{I}_t$, with $\delta \in (0,1)$, $a \geq 0$, $A_1 > 0$ given exogenously, and where $\bar{I}_t$ is the total investment in the economy.\footnote{Intuitively, agents learn new things when they invest effort in production, but cannot keep this new knowledge to themselves. We could of course have used average investment instead. Depreciation is driven by existing technology becoming obsolete or being forgotten.} 

To keep things simple, the population, endowment of raw materials, and productivity of the public good do not change over time. And for tractability, we assume that agents \emph{do not} account for the impact their decisions have on future technology. From a technical standpoint, this shuts down dynamic considerations regarding investment. Agents treat $(A_1,A_2...)$ as exogenous. The key implication of this is that the game is therefore a repeated one-shot stage game, rather than a true dynamic game.

While keeping technological advancement external to agents is standard within classic learning-by-doing models (see for example \cite{arrow1962economic, romer1986increasing} or \cite[Ch.14]{ljungqvist2004recursive}), it is particularly important for tractability here. It is also stronger, in that it requires that agents ignore the impact that both their investment decisions \emph{and} their decisions to extend property rights have on future productivity. The second part of this assumption is perhaps most reasonable if we view productivity as something that grows relatively slowly, so a given member of the elite may not live to see this higher productivity.\footnote{An alternative way of obtaining this feature would be to assume instead that agents are myopic, and only consider the within-period payoffs, or only live for one period and are then replaced by new agents with identical preferences and endowments. Both fit a world where the elite may not live to see the impact their actions today on future productivity.}

\subsection{Results}
Agents end up playing the game \emph{as if} each stage game were one-shot (even if they are forward looking and patient). Behaviour is the same in each period: the elite extends the franchise and chooses inclusive institutions if and only if the threshold in \Cref{eq:threshold} is met, and the disenfranchised invest if and only if the elite will choose inclusive institutions. This does not rely on backwards induction logic. Instead, it is because the one strategic interaction between periods -- namely that extending property rights at period $t$ affects what can happen at all $t' > t$ -- does not have any bite. Hence 

\begin{prop}\label{prop:repeated_game}
    There exists a unique Subgame Perfect Nash Equilibrium. In each period $t=1,2,...$, all agents play as in \Cref{prop:eq}.
\end{prop}

\begin{proof}
As in \Cref{prop:eq}, it can only be optimal to not extend property rights (and keep $E_t = E^0_1$), or extend them to exactly $E_t = 1/G$.
Suppose the elite were totally myopic. And consider a period $t$ where the elite have not yet extended property rights (so $E_t^0 = E_1^0$). Then they would extend property rights to $E_t = 1/G$ if and only if 
\begin{align*}
E^0_1 \left( G + (G - \frac{1}{N}) \frac{A_t f(I_{d,t}^*)}{M} \right) \geq 1,
\end{align*}
as in \Cref{prop:eq}. Now notice that $A_t$ can only increase over time -- $A_{t'} \geq A_t$ for all $t' > t$. Therefore, if the inequality above holds in period $t$, it must also hold in $t'$ for all $t' > t$. So if, ignoring all future periods, the elite want to extend property rights in period $t$, then they would also want to do so in $t+1$.

And recall that agents treat the productivity parameters $A_1,A_2,...$ as exogenous by assumption. Therefore a patient agent will extend property rights if and only if a myopic agent would also do so. 
Investment decisions ($I_{d,t}^*, I_{e,t}^*$) and institutional arrangements ($V_t \in \{S,P\}$) are made anew each period, and do no impact any other period $t' \neq t$. Therefore they follow directly from the choice of the property rights extension, $E_t$. 
\end{proof}



For a given size of the elite, the economy will converge to a steady state level of technology and hence of output. However, there are two different values the franchise can take on; $E^0_1$ (where the initial elite have never extended property rights) or $1/G$ (where property rights have been extended to resolve the commitment problem). This then yields two different steady states for technology.
\begin{defn}[Two steady states]\label{defn:steady_states}
\begin{align*}
    A^{ss, L} \text{ solves } A &= \frac{a}{\delta} E^0 f^{'-1}(1/A)    \tag{Low steady state} \\
    A^{ss, H} \text{ solves } A &= \frac{a}{\delta} \left[ \frac{1}{G} f^{'-1}(1/A) + (N-\frac{1}{G}) f^{'-1}(1/GA) \right] \tag{High steady state}
\end{align*}
\end{defn}
Steady state payoffs and total output follow immediately from this, plus the solution to the game from Propositions \ref{prop:eq} and \ref{prop:repeated_game}. From this, the long-run effects of the initial size of the elite and the initial endowment of raw materials become apparent. 
When the low steady state value of the technology parameter (combined with the other exogenous variables) is not sufficient to induce an extension of property rights, then the economy gets stuck in the low output steady state. 

Suppose an economy starts off with a very low level of technology, $A_1$. Inducing investment by the disenfranchised will not be very attractive -- with rudimentary technology, spending effort does not result in much production of finished goods. So the elite prefer to maintain extractive institutions and expropriate what raw materials are available. Nevertheless, there is some investment (by the elite) and so some productivity growth. But this growth is lower than if there were inclusive institutions, because there is more learning-by-doing when all agents, rather than just the elite, invest.

This growth in productivity increases the attractiveness of inducing investment by the disenfranchised. For some economies, productivity will grow to a point where it is then worthwhile for the elite to extend property rights and provide public goods in order to induce investment by the disenfranchised. Following this improvement in institutional arrangements, productivity growth will increase -- a consequence of the higher overall investment.
However, for other economies, productivity will converge to a steady state -- and so growth will stop -- without it being in the elite's interest to extend property rights and provide inclusive institutions.  

\begin{rem}\label{prop:poverty trap} 
An economy converges to the low steady state if
\begin{align*}
    E^0 \left( G + \left( G - \frac{1}{N} \right) \frac{A^{ss,L} f(I_n^*)}{R} \right) < 1.
\end{align*}
Otherwise it converges to the high steady state.
\end{rem}

\begin{proof}
    This follows immediately from \Cref{prop:eq}, \Cref{prop:repeated_game} and \Cref{defn:steady_states}.
\end{proof}

A resource curse -- a large endowment of natural resources creating worse long-run outcomes -- is subsumed within this analysis. A large endowment of natural resources (`raw materials' in our model) makes extending property rights more costly. In our model, problem associated with a large endowment of natural resources is two-fold. First, it can keep an economy stuck in the low steady state. That is, it never becomes worthwhile for the elite to choose inclusive institutions (i.e. extend property rights and provide public goods). Second, even if it does not prevent the transition to the better institutional arrangement, it will delay the transition. The logic is identical. With more natural resources, extractive institutions are more attractive to the elite. So a higher level of productivity $A_t$ is needed to induce the transition. 

An economy needs to reach some `escape velocity' in order for the elite to want to transition to more inclusive institutions. While not part of the formal analysis, it is clear that temporary shocks to the initial size of the elite or (perhaps more plausibly) to the endowment of raw materials can have a long run impact via their impact on the choice of institutions.



%


\end{document}